\documentclass[oneside,english]{amsart}
\usepackage{mathpazo}
\usepackage[T1]{fontenc}
\usepackage[latin9]{inputenc}
\usepackage{babel}
\usepackage{units}
\usepackage{amsthm}
\usepackage{graphicx}
\usepackage[all]{xy}
\usepackage[unicode=true,pdfusetitle,
 bookmarks=true,bookmarksnumbered=false,bookmarksopen=false,
 breaklinks=false,pdfborder={0 0 0},backref=false,colorlinks=false]
 {hyperref}
\usepackage{breakurl}

\makeatletter
\numberwithin{equation}{section}
\numberwithin{figure}{section}
\theoremstyle{plain}
\newtheorem{thm}{\protect\theoremname}[section]
  \theoremstyle{plain}
  \newtheorem{prop}[thm]{\protect\propositionname}
  \theoremstyle{definition}
  \newtheorem{defn}[thm]{\protect\definitionname}
  \theoremstyle{plain}
  \newtheorem{lem}[thm]{\protect\lemmaname}

\usepackage{xypic}

\makeatother

  \providecommand{\definitionname}{Definition}
  \providecommand{\lemmaname}{Lemma}
  \providecommand{\propositionname}{Proposition}
\providecommand{\theoremname}{Theorem}

\begin{document}
\global\long\def\d{\mathrm{d}}
\global\long\def\R{\mathbf{R}}
\global\long\def\N{\mathbf{N}}
\global\long\def\Z{\mathbf{Z}}
\global\long\def\Q{\mathbf{Q}}
\global\long\def\C{\mathbf{C}}
\global\long\def\P{\mathbf{P}}
\global\long\def\cP{\mathcal{P}}
\global\long\def\cM{\mathcal{M}}
\global\long\def\cL{\mathcal{L}}
\global\long\def\cO{\mathcal{O}}

\global\long\def\pder#1#2{\frac{\partial#1}{\partial#2}}
\global\long\def\restr#1{|_{#1}}
\global\long\def\defeq{:=}
\global\long\def\qefed{=:}
\global\long\def\id{\mathrm{id}}
\global\long\def\dual{\vee}

\global\long\def\fW{\mathfrak{W}}
\global\long\def\fU{\mathfrak{U}}
\global\long\def\fP{\mathfrak{P}}
\global\long\def\fL{\mathfrak{L}}
\global\long\def\CZ{\mu_{\mathrm{CZ}}}
\global\long\def\bH{\mathbf{H}}
\global\long\def\bF{\mathbf{F}}
\global\long\def\e{\mathrm{e}}
\global\long\def\ri{\mathrm{i}}

\global\long\def\modspace{\overline{\mathcal{M}}}
\global\long\def\bh{\mathbf{h}}
\global\long\def\bff{\mathbf{f}}
\global\long\def\ev{\mathrm{ev}}
\global\long\def\bu{\mathbf{u}}
\global\long\def\circint{\frac{1}{2\pi}\intop_{0}^{2\pi}}
\global\long\def\varcircint{\intop_{S^{1}}}

\global\long\def\Lnorm#1{\left\Vert #1\right\Vert _{L^{2}}}
\global\long\def\Wnorm#1{\left\Vert #1\right\Vert _{W^{1,2}}}
\global\long\def\tr{\operatorname{Tr}}
\global\long\def\res{\operatorname{Res}}
\global\long\def\Lie{\operatorname{Lie}}
\global\long\def\diag{\operatorname{diag}}

\global\long\def\Lnormdom#1#2{\Vert#1\Vert_{L^{2}(#2)}}
\global\long\def\Wnormdom#1#2{\Vert#1\Vert_{W^{1,2}(#2)}}

\global\long\def\fg{\mathfrak{g}}
\global\long\def\fh{\mathfrak{h}}
\global\long\def\fm{\mathfrak{m}}
\global\long\def\GL{\operatorname{GL}}
\global\long\def\SL{\operatorname{SL}}
\global\long\def\Ad{\operatorname{Ad}}
\global\long\def\ad{\operatorname{ad}}
\global\long\def\aut{\operatorname{Aut}}

\author{Timo Kluck}

\email{T.J.Kluck@uu.nl}

\address{Department of Mathematics, Utrecht University}

\title{On the Calogero-Moser solution by root-type Lax pair}
\begin{abstract}
The `root type Lax pair' for the rational Calogero-Moser system for
any simply-laced root system yields not a solution for the path $q(t)$,
but for the values of the inner products $(\alpha,q(t))$, where $\alpha$
ranges over all roots of the root system. It does not, however, tell
us which value of the inner product corresponds to which root. In
the present paper, we show that the solution is indeed uniquely determined
by these values (up to root system automorphisms) \emph{at almost
all times}. We show by counterexample that it is possible for two
different values of $q$ to yield the same set of values for the inner
products $(\alpha,q)$.

The indeterminacy introduced by the root system automorphisms introduces
the interesting question when the path crosses from one fundamental
domain into another. We present an algebraic approach for constructing
an indicator function containing this information.
\end{abstract}
\maketitle

\section{Introduction}

The rational Calogero-Moser system is a system of a finite number
of particles on a line, whose pairwise interaction potential at distance
$d$ is given by $\nicefrac{1}{d^{2}}$. Given certain initial values
for positions $q=(q_{i})$ and momenta $p=(p_{i})$, one is interested
in finding the coordinates at later times. There is a remarkable way
of solving this: it turns out that the coordinates at time $t$ are
given by the eigenvalues of the matrix
\begin{equation}
W_{0}+tL_{0}\label{eq:W-plus-t-times-L}
\end{equation}
where $W_{0}$ and $L_{0}$ are constructed from the initial values
by 
\begin{eqnarray*}
W_{0} & = & \left(\begin{array}{ccc}
q_{1}\\
 & \ddots\\
 &  & q_{n}
\end{array}\right)\\
L_{0} & = & \left(\begin{array}{ccc}
p_{1} &  & \frac{1}{q_{i}-q_{j}}\\
 & \ddots\\
\frac{1}{q_{i}-q_{j}} &  & p_{n}
\end{array}\right)
\end{eqnarray*}
The particles' paths cannot cross because their interaction potential
is infinite when they meet; so at a given time, the unordered set
of eigenvalues can be ordered from smallest to greatest to obtain
the positions of each particle.

It was observed by Olshanetsky and Perelomov \cite{olshanetsky-} that
this method of solution depends crucially on the property that the
set of linear maps $q\mapsto q_{i}-q_{j}$ forms a root system (namely
the $A_{n}$ root system in the case of $n+1$ particles), and that
similar methods of solution work for systems whose interaction potential
is given by
\[
\frac{1}{2}\sum_{\alpha\in\Phi}\frac{1}{(\alpha,q)^{2}}
\]
for other root systems $\Phi$. In \cite{Bordner-Calo} and \cite{sasaki-Expli},
Bordner, Corrigan, Sasaki and Takasaki introduce a a Lax pair that
works for all irreducible simply-laced root systems. Their matrices
are much bigger, having a row and column for each root in the root
system, and are given by
\begin{eqnarray*}
\left(W_{0}\right)_{\alpha,\beta} & = & \delta_{\alpha,\beta}\cdot(\alpha,q_{0})\\
\left(L_{0}\right)_{\alpha,\beta} & = & \delta_{\alpha,\beta}\cdot(\alpha,p_{0})+\ri\cdot\sum_{\eta\in\Phi}\frac{\delta_{\alpha-\beta,\eta}}{(\eta,q_{0})}+\frac{2\delta_{\alpha-\beta,2\eta}}{(\eta,q_{0})}
\end{eqnarray*}

Again defining $W(t)=W_{0}+tL_{0}$, the set of eigenvalues $\Lambda(t)$
of $W(t)$ turns out to be equal to the multi-set%
\footnote{By a multi-set, we mean a set $X$ together with a map $\mu\colon X\to\Z_{>0}$,
where we interpret the value of $\mu$ as a multiplicity. When it
is clear from the context, we may drop the ``multi'' prefix. %
} of real numbers
\begin{equation}
M(t):=\{(\alpha,q(t))\mid\alpha\in\Phi\}\label{eq:M-is-alpha-q}
\end{equation}
where $q(t)$ is the path of the position coordinates, and $(\cdot,\cdot)$
is the inner product in the ambient Euclidean space $F$ of $\Phi$.
Therefore, once these eigenvalues are known (step 1), and once we
know which eigenvalue corresponds to which root $\alpha$ (step 2),
all that is left is to solve a system of linear equations for $q$
(step 3). It is the second of these steps that presently interests
us.

First of all, it is clear that the set $M(t)$ can only determine
$q(t)$ up to isometries of $F$ that leave $\Phi$ invariant; that
is, up to root system automorphisms. Therefore, we immediately see
that the correspondence is not uniquely defined. Our first task is
to show that this is the only indeterminacy. It turns out that this
is true only generically; we will show, by a counterexample in the
$\Phi=\Phi(A_{5})$ case, that there can be points $q,q^{\prime}$
in $F$ that are in distinct automorphism orbits, but that nevertheless
have
\[
\{(\alpha,q)\mid\alpha\in\Phi\}=\{(\alpha,q^{\prime})\mid\alpha\in\Phi\}
\]

Our next objective is the following. Since we can solve for $q(t)$
only up to the group action, one could say that the natural domain
for $q$ is $F/\aut(\Phi)$ instead of $F$. However, for a physical
system of particles, this is quite unsatisfactory, as we in general
do distinguish initial values even when they are in the same $\aut(\Phi)$-orbit.
(For instance, in the $A_{n}$ case, we do want to distinguish a solution
from its mirror image.) Therefore, we should divide $F$ into fundamental
domains, and we should find out, for a given initial value, at what
time the path will cross the boundary of the fundamental domains.
This, together with the initial value, allows us to resolve the ambiguity
and reconstruct the path $q(t)$ in its entirety.

This task becomes more interesting when we require the following.
Note that finding the eigenvalues of a given matrix involves finding
the zeroes of a polynomial, which cannot, in general, be done in closed
form. We therefore require that we formulate our answer in terms of
the coefficients of the characteristic polynomial $\chi(W(t))$ for
$W(t)$ (for which we do have explicit formulae) and not in terms
of the eigenvalues of $W(t).$

\section{Example: The $A_{2}$ case}

\begin{figure}
\includegraphics[height=0.4\paperwidth]{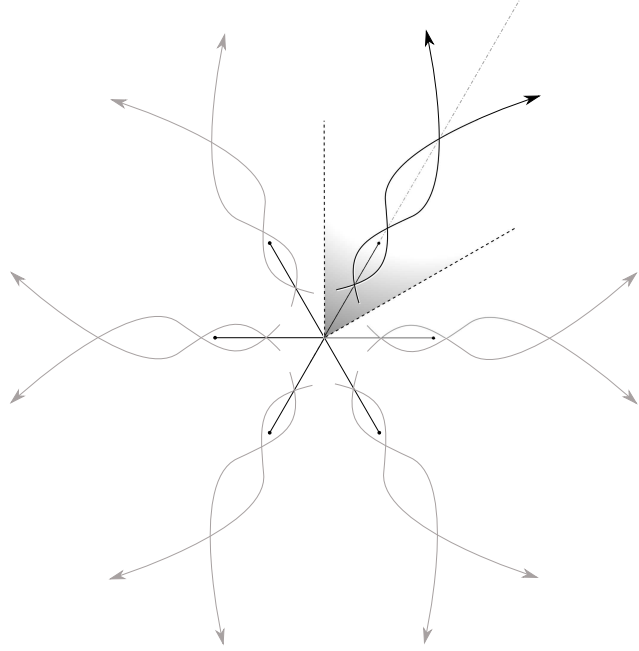}\label{fig:a2-case}

\caption{A set of paths $q(t)$ that are in the same $\aut(\Phi)$-orbit, for
$\Phi=\Phi(A_{2})$. The shaded area is a Weyl chamber. The Dynkin
diagram automorphism acts on the Weyl chamber by reflection though
the dotted line in the middle.}
\end{figure}

Let us discuss the problem that we are trying to solve in the case
where $\Phi$ is the root system associated to $A_{2}$. This root
system is most naturally descibed in the Euclidean subspace $F\subseteq\R^{3}$
satisfying $q_{1}+q_{2}+q_{3}=0$. The root system is given by the
set of vectors $\pm(1,-1,0)$, $\pm(0,1,-1)$, $\pm(1,0,-1)$ in $F$,
and the interaction term in the Hamiltonian is
\[
\frac{1}{2}\sum_{\alpha\in\Phi}\frac{1}{(\alpha,q)}=\sum_{\begin{array}{c}
{\scriptstyle i,j=1}\\
{\scriptstyle i<j}
\end{array}}^{3}\frac{1}{\left(q_{i}-q_{j}\right)^{2}}
\]
The Weyl group is $S_{3}$ and it acts by permuting the coordinates
(which clearly leaves the Hamiltonian invariant), and the nontrivial
Dynkin diagram automorphism acts by sending $q\mapsto-q$. Together,
they generate the root system automorphisms $\aut(\Phi)$. Note that
a Weyl chamber is a fundamental domain for the action of the Weyl
group, whereas either half of the Weyl chamber is a fundamental domain
for the action of the entire $\aut(\Phi)$.

Figure \ref{fig:a2-case} depicts a path $q(t)$ in $F$ and the other
paths in its orbit. The orbit consists of six paths because of the
$S_{3}$ group action alone, and this number is doubled by the Dynkin
diagram automorphism. It is clear that at a given time $t$, all these
values for $q(t)$ yield the same set for $M=\{(\alpha,q)\mid\alpha\in\Phi\}$.
In the $A_{2}$ case, it is easy to see that conversely, two values
giving the same set are in the same $\aut(\Phi)$-orbit: after choosing
a Weyl chamber, the maximal value in $M$ must be associated to the
maximal root. After this, there are only two positive values left,
which can be assigned in exactly two ways to the two positive roots.
In larger root systems, this converse is not so obvious -- in fact,
it is false in general. We will discuss this issue in section \ref{sec:indeterminacy}.

The shaded area is a Weyl chamber, given by $q_{1}<q_{2}<q_{3}$.
It is easy to distinguish paths in different Weyl chambers, because
the Hamiltonian is infinite along the borders; if the initial value
of a path is in a particular Weyl chamber, the path will stay there
for all time. However, this does not allow us to distinguish the two
paths in the same Weyl chamber, related by $(q_{1},q_{2},q_{3})\mapsto(-q_{3},-q_{2},-q_{1})$.
We will tackle this problem as follows. First, we identify a hyperplane
that separates one fundamental domain from the other. In the $A_{2}$
case, this hyperplane is just the line of fixed points, but in general,
the fixed points are only contained in this hyperplane%
\footnote{Here, we restrict to root systems having exactly 2 Dynkin diagram
automorphisms. The case of only a single automorphism is trivial,
so we only exclude the $D_{4}$ case.%
}. Next, we try to identify the times at which the paths cross this
hyperplane. (In general, the paths do not need to intersect each other
as they do in the $A_{2}$ case.) The result obtained in section \ref{sec:Fundamental-domain-crossings}
will be a real polynomial with zeroes exactly where this happens.
Assuming these are simple zeroes, this means that this polynomial
takes positive values at times where we should take one fundamental
domain, and negative values when we should take the other.

\section{Indeterminacy of the solution\label{sec:indeterminacy}}

Let $\Phi\subseteq F$ be an irreducible, simply-laced root system
in a Euclidean space $F$. Simply-laced means that all roots $\alpha$
have $(\alpha,\alpha)=2$; it can be shown that this implies the following
relations that we will use:
\begin{itemize}
\item if $(\alpha,\beta)=-1$, then $\alpha+\beta\in\Phi$;
\item if $(\alpha,\beta)=+1$, then $\alpha-\beta\in\Phi$;
\item otherwise (i.e.\ if $(\alpha,\beta)\in\{-2,0,2\}$), we have $\alpha\pm\beta\notin\Phi$.
\end{itemize}
Suppose that we are given a multi-set $\Lambda$ of real numbers,
and we know that it is equal to some $M$ of the form \eqref{eq:M-is-alpha-q}
(we will drop the time-dependence in our notation in this section).
Another way to say this is that there is a bijection $\phi\colon\Phi\to\Lambda$
such that 
\begin{equation}
\phi(\alpha)=(\alpha,q)\mbox{ for all }\alpha\in\Phi\label{eq:condition-for-matching}
\end{equation}
Here, the word `bijection' should be interpreted as: $\phi$ is a
map from $\Phi$ to the underlying set of $\Lambda$, such that the
size of the preimage of each point is equal to its multiplicity in
$\Lambda$. Given such a bijection, we can solve for $q$. In fact,
$q$ is already determined by its inner products with the simple roots
$\Delta\subseteq\Phi$, since these are $\dim F$ linear equations
for $\dim F$ unknowns.

It is clear that if $\sigma\in\aut(\Phi)$ is a root system automorphism,
then $\phi\circ\sigma$ will be another bijection that satisfies \eqref{eq:condition-for-matching}.
The converse needs proof:
\begin{prop}
\label{prop:main-prop}The following statement is true for generic
$\Lambda$: If $\phi_{1}$ and $\phi_{2}$ are two bijections $\Phi\to\Lambda$
satisfying \eqref{eq:condition-for-matching}, then there is an automorphism
$\sigma\in\aut(\Phi)$ such that $\phi_{1}\circ\sigma=\phi_{2}$.
\end{prop}
Note that in the case where $\Lambda$ has multiple values, the condition
$\phi_{1}\circ\sigma=\phi_{2}$ does not even fix $\sigma$ as a bijection,
so it is not an entirely trivial matter to find a suitable $\sigma$.
The proof will need the following definition and lemma.
\begin{defn}
Let $\sigma\colon\Phi\to\Phi$ be any map. We call $\sigma$ \emph{additive}
if it satisfies these conditions:
\[
\sigma(\alpha)+\sigma(\beta)=\sigma(\alpha+\beta)\mbox{ }
\]
for any value of $\alpha,\beta\in\Phi$ such that $\alpha+\beta$
is a root, and
\[
-\sigma(\alpha)=\sigma(-\alpha)
\]
for any $\alpha\in\Phi$.

It is important to realize that $\Phi$ is not a group under addition.
In particular, this means that the second condition does not follow
from the first.\end{defn}
\begin{lem}
\label{lem:additivity-of-sigma}The following statement is true for
generic $\Lambda$. Suppose $\sigma$ is a bijection $\Phi\to\Phi$
such that $\phi_{1}\circ\sigma=\phi_{2}$. Then $\sigma$ is additive.\end{lem}
\begin{proof}
First of all, note that from \eqref{eq:condition-for-matching}, it
follows in particular that the $\phi_{i}$ are additive, in the sense
that $\phi_{i}(\alpha)+\phi_{i}(\beta)=\phi_{i}(\alpha+\beta)$ and
$\phi_{i}(-\alpha)=-\phi_{i}(\alpha)$. However, this does not imply
that $\phi_{i}^{-1}$ is additive: it is possible that 
\[
\phi_{i}(\alpha)+\phi_{i}(\beta)=\lambda_{1}+\lambda_{2}=\lambda\in\Lambda
\]
 even if $\alpha+\beta$ is not a root. In this case, $\phi_{i}^{-1}(\lambda_{1})+\phi_{i}^{-1}(\lambda_{2})$
cannot equal $\phi_{i}^{-1}(\lambda)$.

We write $\Z^{\Phi}$ for the free abelian group with a set of generators
indexed by $\Phi$. There is a canonical map $\pi\colon\Z^{\Phi}\to\Z\cdot\Phi$
to the root lattice, whose kernel $\ker\pi$ contains exactly the
additivity relations. The map $\sigma$ induces a map $\sigma_{*}\colon\Z^{\Phi}\to\Z^{\Phi}$.
It is easy to see that $\sigma$ is additive if and only if $\sigma_{*}$
maps $\ker\pi$ to itself, in other words, if and only if
\[
\ker\pi\subseteq\ker\pi\circ\sigma_{*}
\]

Now let us consider the $\phi_{i}$. We see that each extends linearly
to a map $\phi_{i*}\colon\Z\cdot\Phi\to\R$ (we use their additivity
here). The statement that $\phi_{1}\circ\sigma=\phi_{2}$ implies
that we have the following commutative diagram:
\[
\xymatrix{\Phi\ar[d]^{\sigma}\ar@{^{(}->}[r] & \Z^{\Phi}\ar[r]^{\pi}\ar[d]^{\sigma_{*}} & \Z\cdot\Phi\ar[r]^{\phi_{1*}}\ar@{-->}[d] & \R\ar@{=}[d]\\
\Phi\ar@{^{(}->}[r] & \Z^{\Phi}\ar[r]^{\pi} & \Z\cdot\Phi\ar[r]^{\phi_{2*}} & \R
}
\]
The dotted arrow is a map that exists if and only if $\sigma$ is
additive.

We see from the diagram that 
\[
\ker\phi_{1*}\circ\pi=\ker\phi_{2*}\circ\pi\circ\sigma_{*}
\]
 This means that
\begin{eqnarray*}
\ker\pi & \subseteq & \ker\phi_{1*}\circ\pi\\
 & = & \ker\phi_{2*}\circ\pi\circ\sigma_{*}\mbox{ (by our observation)}\\
 & = & \ker\pi\circ\sigma_{*}+\sigma_{*}^{-1}\circ\pi^{-1}\left(\ker\phi_{2*}\right)
\end{eqnarray*}
so it is sufficient if we can prove that $\ker\phi_{2*}$ is trivial.

Note that $\ker\pi$ is generated by linear combinations of at most
3 generators. This means that it is actually sufficient to show that
$\sigma_{*}^{-1}\circ\pi^{-1}\left(\ker\phi_{2*}\right)$ does not
contain elements that small. $ $In fact, $\sigma_{*}$ preserves
norms and $\pi$ only makes them smaller, so it is sufficient if $\ker\phi_{2*}$
does not contain elements of length smaller than $\sqrt{3}$.

Now, $\ker\phi_{2*}$ is a hyperplane of codimension 1 in the ambient
space $F$ of the root lattice $\Z\circ\Phi$. For generic values
of $\phi_{2}$'s coefficients $\Lambda$, this hyperplane has trivial
intersection with 
\[
\Z\cdot\Phi\cap\{x\in F\mid\left|x\right|\le\sqrt{3}\}
\]
This means that 
\[
\ker\pi\subseteq\ker\pi\circ\sigma_{*}
\]
 for generic $\Lambda$. This implies that $\sigma$ is additive.
\end{proof}

\begin{proof}[Proof of proposition \ref{prop:main-prop}]

Choose any $\sigma$ that satisfies $\phi_{1}\circ\sigma=\phi_{2}$,
and choose a a base $\Delta\subseteq\Phi$ for the root system. There
is a unique linear map $L\colon F\to F$ that extends $\sigma$ on
$\Delta$. Because $\sigma$ is additive, we see that $L$ actually
extends $\sigma$ on all of $\Phi$. We will now show that $L$ is
an isometry. We can check this on $\Delta$ because these roots span
$F$. Because $L$ extends a permutation of the roots, it is clear
that $L\alpha\pm L\beta\in\Phi$ if and only if $L(\alpha\pm\beta)\in\Phi$.
But in light of the relations above, this means that 
\[
(L\alpha,L\beta)=(\alpha,\beta)
\]
so $L$ is an isometry and $\sigma$ is an automorphism of $\Phi$.
\end{proof}
Now remember that the Hamiltonian has a pole along $(\alpha,q)=0$
for every root $\alpha$, that is, along the boundaries of the Weyl
chambers. This means that if the initial value for $q$ is in a certain
Weyl chamber, then it will stay there for all $t$. This means that
we can fix a set of positive roots corresponding to the Weyl chamber,
and use an element of the Weyl group $w\in W\subseteq\aut(\Phi)$
to make sure our bijection sends positive roots $\alpha$ to positive
values of $(q,\alpha)$. This will make sure that the corresponding
solution for $q$ is in the right Weyl chamber. This proves the following
\begin{prop}
For generic $\Lambda$, there are exactly $[\aut(\Phi):W]$ possibilities
for the solutions for $q$. They are related by the action of the
Dynkin diagram automorphisms on $F$.
\end{prop}

\section{A non-generic counter example}

The question remains whether the proof of lemma \ref{lem:additivity-of-sigma}
can be made to work for all, instead of just generic, $\Lambda$.
The answer to this question is negative, at least for the case $\Phi=\Phi(A_{5})$.
In this case, consider the vectors
\begin{eqnarray*}
q & = & (-28,-22,-16,8,20,38)\\
q^{\prime} & = & (-34,-28,2,8,20,32)
\end{eqnarray*}
which are not mirror-images and therefore not related by a Dynkin
diagram automorphism. Then we can check that
\[
\{(\alpha,q)\mid\alpha\in\Phi(A_{5})\}=\{(\alpha,q^{\prime})\mid\alpha\in\Phi(A_{5})\}
\]
For reference, both are equal to the multi-set
\begin{equation}
\pm\{6,6,12,12,18,24,30,30,36,36,42,48,54,60,66\}\label{eq:alpha-q-counter-example}
\end{equation}

This counter-example was produced by computer-search%
\footnote{The search was conducted using the open-source software Sage \cite{Sage}.
Source code for the search program is available from the author's
website.%
}. A non-exhaustive search for counterexamples in other small root
systems $D_{4}$,$D_{5}$,$A_{6}$ and $E_{8}$ did not yield any
other examples.

A natural question is how the path evolves when we take these two
points as initial values. The two different $W_{0}$ are diagonal
matrices that have the numbers \eqref{eq:alpha-q-counter-example}
as eigenvalues, but on different rows. There is no similar relation
between the two different $L_{0}$, because the map on $\Phi$ induced
by the correspondence does not respect additivity relations. Therefore,
there is no reason to suspect that the two solutions remain related
in any way. This is illustrated in figure \ref{fig:Time-evolution}.
\begin{figure}
\includegraphics{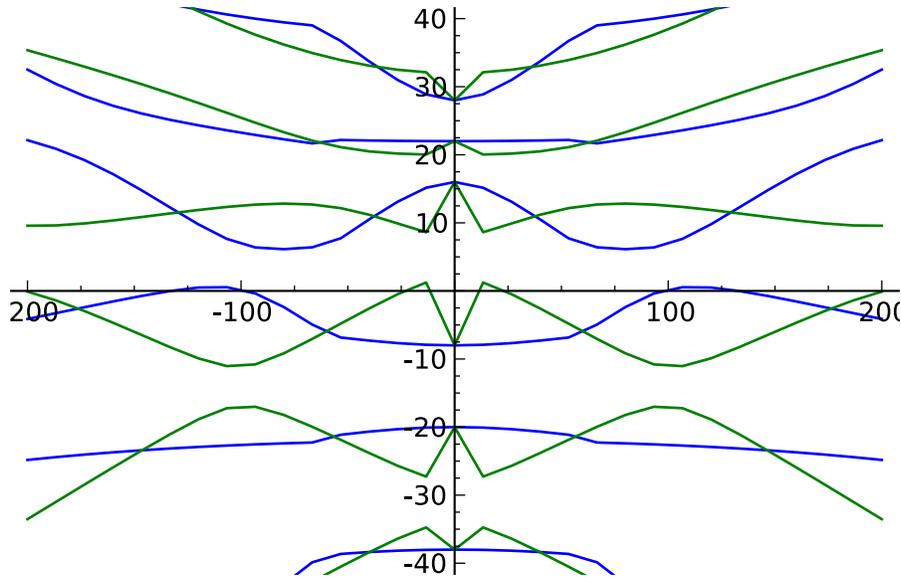}

\caption{Time evolution of the two systems, as calculated by the algorithm
described in this article. At time $t=0$, the initial values $q$
(blue) and $q^{\prime}$ (green) are chosen, together with $p=0$.
The algorithm clearly makes a mistake at $t=0$ where it picks the
`wrong' matching between eigenvalues and roots, namely the one leading
to the other initial condition. The plot is symmetric under $t\mapsto-t$
because when $p=0$, the two matrices $W(t)$ and $W(-t)$ are Hermitean
conjugates and therefore have the same eigenvalues.\label{fig:Time-evolution}}
\end{figure}

\section{Fundamental domain crossings\label{sec:Fundamental-domain-crossings}}

Let us now exclude the cases $E_{7},E_{8}$ (which only have the trivial
diagram automorphism) and the case $D_{4}$ (which has 6 automorphisms)
so that we have exactly 2 solutions for $q$ at every time $t$. In
other words, we have a path in the quotient space 
\[
\mbox{Weyl chamber}/\mbox{Dynkin diagram automorphism}
\]
We would like to separate the two ``lifted'' paths in the Weyl chamber.
One way of doing this is to split the Weyl chamber into two fundamental
domains, and finding out at what time $q(t)$ passes the boundary
from one fundamental domain to another. The Dynkin diagram automorphism
corresponds to a linear isometry of order 2, so we can choose the
fundamental domains as being the two sides of any codimension 1 hyperplane
containing its fixed points%
\footnote{To see this: if two points are on the same side of the hyperplane
and are mapped to each other, then their sum is also on the same side
of the hyperplane, but it is a fixed point.%
}. In fact, we can always take the hyperplane to be the equidistant
hyperplane between two simple roots forming an orbit. For example,
in the $A_{2}$ case, we can choose the boundary of the fundamental
domains to be the dotted line in figure \ref{fig:a2-case}. We then
want to find the times $t$ at which the paths cross the boundary.
We have seen in the $A_{2}$ case that they must cross each other
because the boundary is actually fixed by the automorphism. In general,
however, the boundary need only be mapped into itself, so the paths
can cross the boundary at distinct points.

We would like to express the fact that $q$ is on the boundary by
looking at the values of $(\alpha,q)$. Then because the boundary
is an equidistant hyperplane, we see that there must be double values.
The converse is false, however: there are several more equidistant
hyperplanes, but not all of them are between roots in an orbit, and
from those that are, we have only chosen one as the boundary of the
fundamental domain. Also remember that we want to find a condition
on the coefficients of the characteristic polynomial of $W(t)$; that
is, on the symmetric functions of the $(\alpha,q)$, and not on the
$(\alpha,q)$ themselves.

\subsection{An algebraic condition for the crossing}

There is a standard way of expressing certain conditions on the zeroes
of a polynomial as conditions on coefficients. Let us illustrate this
method by recalling the definition of the discriminant of a polynomial.
Consider a polynomial
\[
\lambda^{m}+a_{m-1}\lambda^{m-1}+\cdots+a_{0}=\prod_{i}(\lambda-\lambda_{i})
\]
We can express the condition that this polynomial has a double zero
(that is, there is $i\neq j$ with $\lambda_{i}=\lambda_{j}$) by
requiring the vanishing of the following expression
\[
\prod_{\sigma\in S_{m}}\left(\lambda_{\sigma(1)}-\lambda_{\sigma(2)}\right)
\]
which is symmetric in the $\lambda_{i}$ and can therefore also be
expressed in the $a_{i}$. This expression is just a power of the
determinant.

Let us apply this to the the condition that $q$ is at the boundary
of a fundamental domain. This is a condition on the zeroes of the
form:
\begin{quote}
There is a way of assigning the $\lambda_{i}$ to roots $\alpha_{\phi(i)}$
such that (1) the $\lambda_{i}$ satisfy the additivity properties
of the roots and (2) they have the same value on two specified roots.
\end{quote}
For example, in the $A_{2}$ case, we want the simultaneous vanishing
of these expressions:

\begin{eqnarray}
\lambda_{1}+\lambda_{2}-\lambda_{3} &  & \mbox{(one root is the sum of two other roots)}\label{eq:additivity-rels-A2}\\
\lambda_{1}-\lambda_{2} &  & \mbox{(those two other roots have the same value) }\nonumber \\
\lambda_{1}+\lambda_{4} &  & \mbox{(the roots have mirror images)}\nonumber \\
\lambda_{2}+\lambda_{5} &  & \mbox{(idem)}\nonumber \\
\lambda_{3}+\lambda_{6} &  & \mbox{(idem)}\nonumber 
\end{eqnarray}
The second of these corresponds to (2), and the others correspond
to (1). Of course, vanishing of a simultaneous permutation is also
allowed, since that corresponds to a different way of assigning roots
to the $\lambda_{i}$. The simultaneous vanishing of any simultaneous
permutation can be encoded in the vanishing of the following expression:
\begin{multline*}
\prod_{\sigma\in S_{6}}\left((\lambda_{\sigma(1)}-\lambda_{\sigma(2)})+y_{1}(\lambda_{\sigma(1)}+\lambda_{\sigma(2)}-\lambda_{\sigma(3)})\right.\\
\left.+y_{2}(\lambda_{\sigma(1)}+\lambda_{\sigma(4)})+y_{3}(\lambda_{\sigma(2)}+\lambda_{\sigma(5)})+y_{4}(\lambda_{\sigma(3)}+\lambda_{\sigma(6)})\right)
\end{multline*}
\emph{identically in the helper variables $y_{1},\cdots,y_{4}$}.
Taking the product over all permutations makes that the condition
is symmetric in the $\lambda_{i}$, allowing us the express it in
the $a_{i}$. Note that the identical vanishing will give one condition
$c_{[y]}(a_{0},\cdots,a_{m})$ for every monomial $[y]$ in $y_{1},\cdots,y_{4}$.

Now in general, suppose the additivity relations (such as \eqref{eq:additivity-rels-A2}
in the $A_{2}$ case) take the form $f_{j}(\lambda_{1},\cdots,\lambda_{m})$
for $j$ in some index set $J$, and suppose that the boundary hyperplane
is equidistant to $\alpha_{1},\alpha_{2}$. Then we are interested
in the vanishing of the following expression:
\begin{equation}
\prod_{\sigma\in S_{m}}\left(\lambda_{\sigma(1)}-\lambda_{\sigma(2)}+\sum_{j\in J}y_{j}f_{j}(\lambda_{\sigma(1)},\cdots,\lambda_{\sigma(m)})\right)\label{eq:extended-determinant}
\end{equation}
identically in the helper variables $y_{j}$\emph{.} Again, this condition
can be expressed in the $a_{i}$, and we obtain one condition $c_{[y]}(a_{0},\cdots,a_{m})$
for every monomial $[y]$ in the $y_{j}$. In our case, the coefficients
$a_{0},\cdots,a_{m}$ are the coefficients of the characteristic polynomial
of $W(t)$. That means that we have explicit formulae $a_{0}=a_{0}(t),\cdots,a_{m}=a_{m}(t)$.
These are polynomials in $t$ with coefficients in $\Q(p_{0},q_{0})$,
where $p_{0}$ and $q_{0}$ are the initial values.

Since $\Q(p_{0},q_{0})[t]$ is a unique factorization domain, there
is a well-defined greatest common divisor $c$ of all the $c_{[y]}(a_{0}(t),\cdots,a_{m}(t))$.
This greatest common divisor vanishes exactly when $q$ is on the
boundary of the fundamental domain at time $t$. This is a real polynomial
with zeroes exactly at boundary crossings. Assuming the zeroes are
simple, we can interpret it as an indicator function that is positive
at times where the path is in one fundamental domain, and negative
when it is in the other.

\subsection{Feasibility of the computation}

The computation just described is unfeasible, even for the smallest
of root systems. In the case $A_{2}$, we have $6$ roots, so the
expression \eqref{eq:extended-determinant} is a homogeneous polynomial
of degree $6!=720$ in the 6 $\lambda_{i}$-variables, which means
it has ${720+5 \choose 5}$ summands, a 13 digit number.

A polynomial whose computation is a lot closer to being within reach
is the discriminant $\delta$ of the characteristic polynomial of
$W(t)$. It is zero exactly at times $t$ when $\Lambda$ has double
values. In particular, it is zero when $q(t)$ crosses a boundary.
We find that $c$ is a factor in $\delta$. Furthermore, because we
have just shown that $c$ is a polynomial in $\Q(p_{0},q_{0})[t]$,
we can obtain $c$ by factoring $\delta$ over $\Q(p_{0},q_{0})$.
When we fix rational values for the initial values $p_{0}$ and $q_{0}$,
this is a factorization over $\Q$ and so it is a finite computation.

As an example of this procedure, let $\Phi=\Phi(A_{2})$ and let the
initial values be given by $q_{0}=(\frac{6}{10},-\frac{1}{10},-\frac{1}{2})$
and $p_{0}=(\frac{1}{10},-\frac{1}{10},0)$. Because these are rational
values, $\Q(p_{0},q_{0})$ is just equal to $\Q$. The charateristic
polynomial of $W(t)$ with these initial conditions can be computed
to be equal to
\begin{multline*}
\lambda^{6}+\left(-\frac{7763475}{11858}\, t^{2}+42\, t-\frac{93}{50}\right)\lambda^{4}+\\
\left(\frac{60271544075625}{562448656}\, t^{4}-\frac{23290425}{1694}\, t^{3}+\frac{49797639}{47432}\, t^{2}-\frac{1953}{50}\, t+\frac{8649}{10000}\right)\lambda^{2}+\\
\left(-\frac{17065397825724953125}{3334758081424}\, t^{6}+\frac{5719079645625}{5021863}\, t^{5}-\frac{16356434361825}{281224328}\, t^{4}\right.\\
\left.-\frac{3061123}{1694}\, t^{3}+\frac{235613523}{2371600}\, t^{2}+\frac{3003}{1250}\, t-\frac{5929}{62500}\right)
\end{multline*}
The discriminant of this polynomial is equal to
\begin{gather*}
k\cdot(t^{3}+\frac{44921}{51450}t^{2}-\frac{121}{1875}t+\frac{121}{218750})^{4}\cdot(\mbox{a large polynomial of degree 6 without real roots})^{3}
\end{gather*}
for some large constant $k$. A numerical approximation of the solution
shows that we expect that the fundamental domain border is crossed
three times. This allows us to identify the factor of degree 3 as
the indicator function that is positive when the solution is in one
fundamental domain, and negative when it is in the other. This is
a less rigorous way because it involves comparison of an exact result
with a numerical approximation, but it is at least feasible.

\section{Conclusion}

We have shown that the solution to the Calogero-Moser system by root-type
Lax pair is complete up to root system automorphism, in the sense
that the data it yields determines the solution at almost all times
(and therefore by continuity at all times); however we have shown
by example that ambiguity can occur at isolated points.

Next, in the cases where there are exactly two different Dynkin diagram
automorphisms (so in all cases different from $D_{4},E_{7},E_{8}$),
we have given a way of distinguishing the two paths in the Weyl chamber
by means of an indicator function, whose construction is entirely
algebraic. However, its computation is infeasible even for tiny examples,
but we have also indicated a less rigorous way of obtaining it in
a much less computationally expensive way.

It is known that in cases different from $E_{8}$, an alternative
Lax pair is available (called the minimal Lax pair in \cite{Bordner-Calo})
that yields not values for $(\alpha,q)$, but for $(\lambda,q)$ where
$\lambda$ runs over the fundamental weights. It would be interesting
to see if similar steps are necessary to ensure that the resulting
data completely fix the solution.

\section{Acknowledgements}

This research was helped by computer exploration using the open-source
mathematical software Sage \cite{Sage} and its algebraic combinatorics
features developed by the Sage-Combinat \cite{sage-combinat} community.
In particular, the work by Mike Hansen, Justin Walker and Nicolas\ M.\ Thiery
on root systems was very helpful.

Source code for an implementation of the above algorithm in Sage is
available from the author's website%
\footnote{This website can be found at \texttt{www.staff.science.uu.nl/\textasciitilde{}kluck103/}%
}.

This research was supported by the Utrecht University program `Foundations
of Science'.

\bibliographystyle{amsplain}
\bibliography{references-tjk}

\bigskip{}

\end{document}